\documentclass[journal,twoside,web]{ieeecolor}
\usepackage{lcsys}
\def\BibTeX{{\rm B\kern-.05em{\sc i\kern-.025em b}\kern-.08em
    T\kern-.1667em\lower.7ex\hbox{E}\kern-.125emX}}
\usepackage{amsmath,amsfonts,amssymb}
\usepackage{algorithmic}
\usepackage{array}
\usepackage{url}
\usepackage{graphicx}
\usepackage{textcomp}
\usepackage{multicol}
\usepackage{xcolor}
\usepackage{tikz}
\usepackage{breqn}

\let\labelindent\relax
\usepackage{enumitem}

\usepackage{pbalance}

\usetikzlibrary{decorations.pathreplacing}
\newtheorem{theorem}{Theorem}
\newtheorem{corollary}{Corollary}
\newtheorem{lemma}{Lemma}
\newtheorem{definition}{Definition}
\newtheorem{remark}{Remark}

\usetikzlibrary{positioning, arrows.meta}
\usepackage{mathtools}
\def\|#1\|{\left\|#1\right\|}
\def\argmax{\mathop{\rm argmax}}
\def\Rr{\mathbb{R}}

\def\bpi{\boldsymbol{\pi}}

\def\mcS{\mathcal{S}}

\def\mcA{\mathcal{A}}
\def\gs{\mathbf{s}}

\def\Expect{\mathbf{E}}

\def\boldPi{{\boldsymbol{\Pi}}}

\def\Bel{\mathcal{B}}
\def\boldpi{\boldsymbol{\pi}}
\def\mcV{\mathbf{V}}
\newcommand{\Ram}[1]{\textcolor{black}{#1}}
\begin{document}
\title{Hierarchical Decentralized Stochastic Control for Cyber-Physical Systems}
\author{Kesav Kaza$^*$, Ramachandran Anantharaman$^{*,\dagger}$ and Rahul Meshram
\thanks{K. Kaza is with the Department of Electrical Engineering and Computer Science, University of Ottawa, Canada (e-mail: kkaza@uottawa.ca), R. Anantharaman is with the Department of Electrical Engineering, Eindhoven University of Technology (TU/e), the Netherlands (e-mail: r.chittur.anantharaman@tue.nl) and R. Meshram is with the Department of Electrical Engineering, Indian Institute of Technology Madras, Chennai, India. (e-mail: {rahulmeshram}@ee.iitm.ac.in). R. Meshram is supported from IITM NFIG Grant and SERB grant Project No EEQ/2021/000812.}
\thanks{$*$ equal contributions, $\dagger$ corresponding author.}}

\maketitle

\begin{abstract}
This paper introduces a two-timescale hierarchical decentralized control architecture for Cyber-Physical Systems (CPS). 
\Ram{ The system consists of a global controller (GC), and $N$ local controllers (LCs). The GC operates at a slower timescale, imposing budget constraints on the actions of LCs, which function at a faster timescale. Applications can be found in energy grid planning, wildfire management, and other decentralized resource allocation problems. We propose and analyze two optimization frameworks for this setting: COpt and FOpt. In COpt, both GC and LCs together optimize infinite-horizon discounted rewards, while in FOpt the LCs optimize finite-horizon episodic rewards, and the GC optimizes infinite-horizon rewards. Although both frameworks share identical reward functions, their differing horizons can lead to different optimal policies. In particular, FOpt grants greater autonomy to LCs by allowing their policies to be determined only by local objectives, unlike COpt. To our knowledge, these frameworks have not been studied in the literature. We establish the formulations, prove the existence of optimal policies, and prove the convergence of their value iteration algorithms. We further show that COpt always achieves a higher value function than FOpt and derive explicit bounds on their difference. Finally, we establish a set of sufficient structural conditions under which the two frameworks become equivalent.}  
\end{abstract}

\section{Introduction}

Cyber-physical systems (CPSs) are large scale engineering systems consisting various computing, communication, and control units. They aim to achieve system performance objectives through the control of complex subsystems, and have applications in energy, healthcare, industrial automation, disaster management, defense, etc. \cite{Kim2012:CPSperspective,Rajkumar2010:CPSrevolution}.

\Ram{In this paper, we study a specific CPS architecture featuring two levels of decision-making: a Global Controller (GC) and multiple Local Controllers (LCs). These two layers reflect distinct functions of the system objective, where the LCs manage individual sub-processes through local decision rules, while the GC coordinates global performance across the entire system based on its collective state. We further assume that the two controllers operate on different timescales. The GC acts on a slower timescale, and its decisions impose budget constraints on the cumulative actions of the LCs, which operate on a faster timescale. We formulate this hierarchical decision problem using the Markov Decision Process (MDP) framework. }

{Consider the following motivational application for such a two-scale hierarchical formulation. Forest fires occur frequently in certain parts of the world and can have severe consequences on the livability of these regions, affecting temperature, water availability, air quality, and more. Consider a scenario in which a central agency is responsible for overseeing forest fire control across multiple regions. Each region also has local incident controllers with limited resources who are directly managing the situation on the ground. Due to logistical constraints, the central agency can actively monitor and allocate additional resources to only a limited number of regions at the same time. The local controllers can be modeled as agents allocating local resources according to a finite-horizon MDP with an objective of minimizing the total cost (to lives and property) incurred. They work at a faster time scale which is in the order of hours. The central agency also has the same objective but has to consider all the regions together, and works at a slower timescale, maybe in order of a half a day to few days. This problem was considered as an application to constrained restless bandits in \cite{kaza2024constrained}, where the authors used a two-state model to represent low or high danger in an area. They also provide a brief literature survey and details on wildfire spread models. However, their model assumes no decision-making agency at the local level and is also restricted by a two-state limitation. In this work, we provide two formulations of a hierarchical stochastic control framework which can be applied to achieve precise control in such scenarios.}
\begin{figure}[h]
\centering
\includegraphics[width = .9\columnwidth]{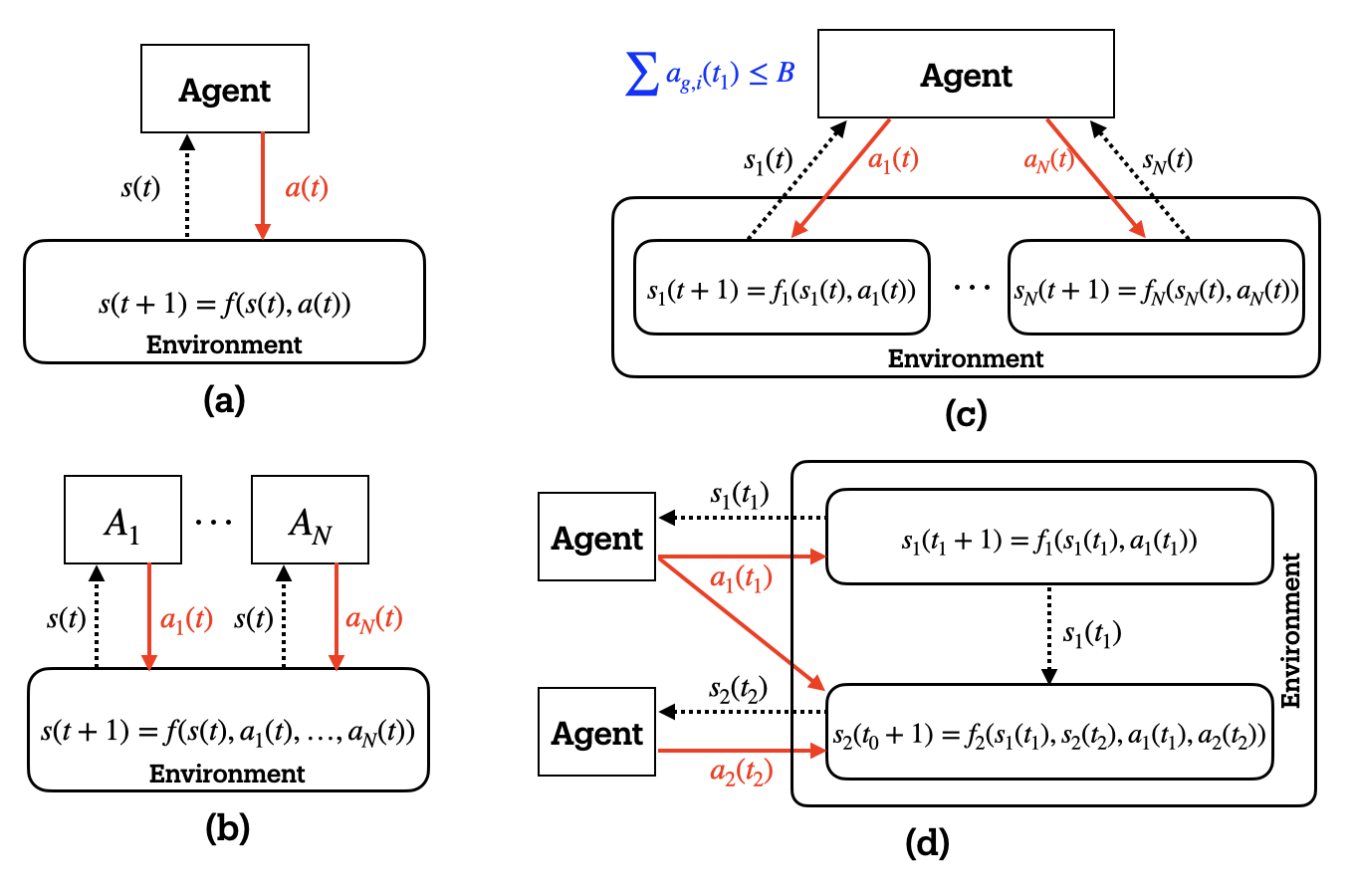}
\caption{A simplistic representation of some major categories of models found in the literature on decentralized control -- (a) single agent, (b) multi-agent--single-timescale, common environment, (c) weakly-coupled, independent sub-processes, (d) multi-timescale ($t_1 \not\equiv t_2$), layered architecture.}
\label{fig:models-literature}
\end{figure}
\subsection{Related work}
\Ram{In this section, we present literature on sequential decision problems with focus on MDP models that are relevant to our work, wherein we divide the relevant literature into four categories as follows. A summary of the models is provided in Fig.~\ref{fig:models-literature}} 

\noindent \textbf{Markov decision processes (MDPs)} are fundamental stochastic control models for sequential decision making. There is a single agent interacting with a environment which generates a state-based reward and the state evolves based on the agent's action according to a Markov process \cite{puterman2014markov}.

\noindent \textbf{Multi-agent models/Decentralized stochastic control (DSC)/Hierarchical control:} In this category, two or more agents interact with a common environment which can be either Markovian or deterministic. The large literature on multi-agent problems includes  decentralized control  \cite{Sandell1975:Survey-Decent-Hierar,Hsu1982:decentralized,mahajan2016decentralized} and team decision (multi-agent) problems with application to stochastic network controlled systems  \cite{ho1980team,yuksel2013stochastic}. 
DSC problems are known to be NP-hard \cite{tsitsiklis2003:complexity}, hence, much effort has been made to find efficient heuristics to solve them \cite{bertsekas2021:multiagent}. 
To analyze the decentralized multi-agent problems, assumptions on information structure are often made. It has been studied extensively in \cite{yuksel2013stochastic}. 
More recently, constrained multi-agent MDPs have been studied for various applications \cite{de2021constrained}.

\noindent \textbf{Weakly coupled MDP/RMAB:} 
 In the context of solving large MDPs efficiently, literature on weakly coupled MDPs (WC-MDPs) had developed with methods such as splitting large state spaces into weakly communicating parts or representing a large MDP as smaller seemingly independent concurrent processes~\cite{Meuleau:AAAI98:WCMDP}. Later, WC-MDPs have emerged as an independent area of study, due to various applications of the model \cite{hawkins2003langrangian, Adelman:2008:Relaxations}.
 The problem of restless multi-armed bandits (RMABs) can be called an independently evolved special case of WC-MDPs \cite{Whittle:1988:RMAB}. Here, the decision maker interacts with multiple independent sub-processes, under the constraint that simultaneously allocate resources to a fixed number of processes. The goal is to choose the optimal subset in each decision interval such that long term cumulative reward is maximized . Several variations of this model have been studied \cite{nino2001restless}, along with applications to dynamic scheduling in CPSs \cite{kaza2024constrained}.

\noindent \textbf{Multi-timescale hierarchical/layered control models:} Two layered multi-timescale model for a Markov and linear control system has been studied in \cite{forestier1978:multilayer}. Here, offline policies were proposed. Much later, \cite{Chang:2003:Multitime} built upon this work, and formulated the multi-timescale MDP. Exact and approximate solutions were discussed along with heuristic online policies.

Our model differs significantly from the above categories of literature. We study a multi-timescale multi-agent hierarchical stochastic control framework (Fig.~\ref{fig:timeline}), where there is a global controller (GC), defined as an MDP, and a set of local controllers (LCs) which are also MDPs managing independent sub-processes. Both of them act on different timescales. The actions of the LCs influence the process state directly. The actions of the global controller act as episodic budget constraints on the LC. Additionally, the GC operates under an allocation constraint. 

\subsection{Our contributions}
The main contributions of this paper are highlighted below. 
\begin{enumerate}[leftmargin=*]
\item We propose central and federal optimization frameworks for two-timescale, multi-agent hierarchical decentralized stochastic systems, and demonstrate the differences between the two frameworks. 
\item We present a theoretical analysis for these frameworks, establish several relationships between the optimal value functions, and bound the difference between them. 
\item Finally, we present a set of sufficiency conditions under which the two frameworks achieve the same value.

\end{enumerate}

The paper is organized as follows. Section \ref{sec:SysDes} gives the system description, Section \ref{sec:Optframeworks} discusses the Central and Federal Optimization frameworks. Section \ref{sec: analysis} is dedicated to the analysis of the two frameworks, followed by the conclusion.

\section{System description and problem formulation}
\label{sec:SysDes}
Consider a hierarchical control system with $N$ local controllers that act on $N$ sub-processes, and a global controller that acts on the LCs. The GC operates at a slower timescale, indexed by $t_1$, while the LCs operate at a faster timescale indexed by $t_0$, with a time scale factor $T$, i.e., $t_0 = Tt_1 + t$, $t = 0,\dots,T-1$, as shown in Fig.~\ref{fig:timeline}(b).

\begin{figure}[t]
\begin{center}
\includegraphics[width = .85\columnwidth]{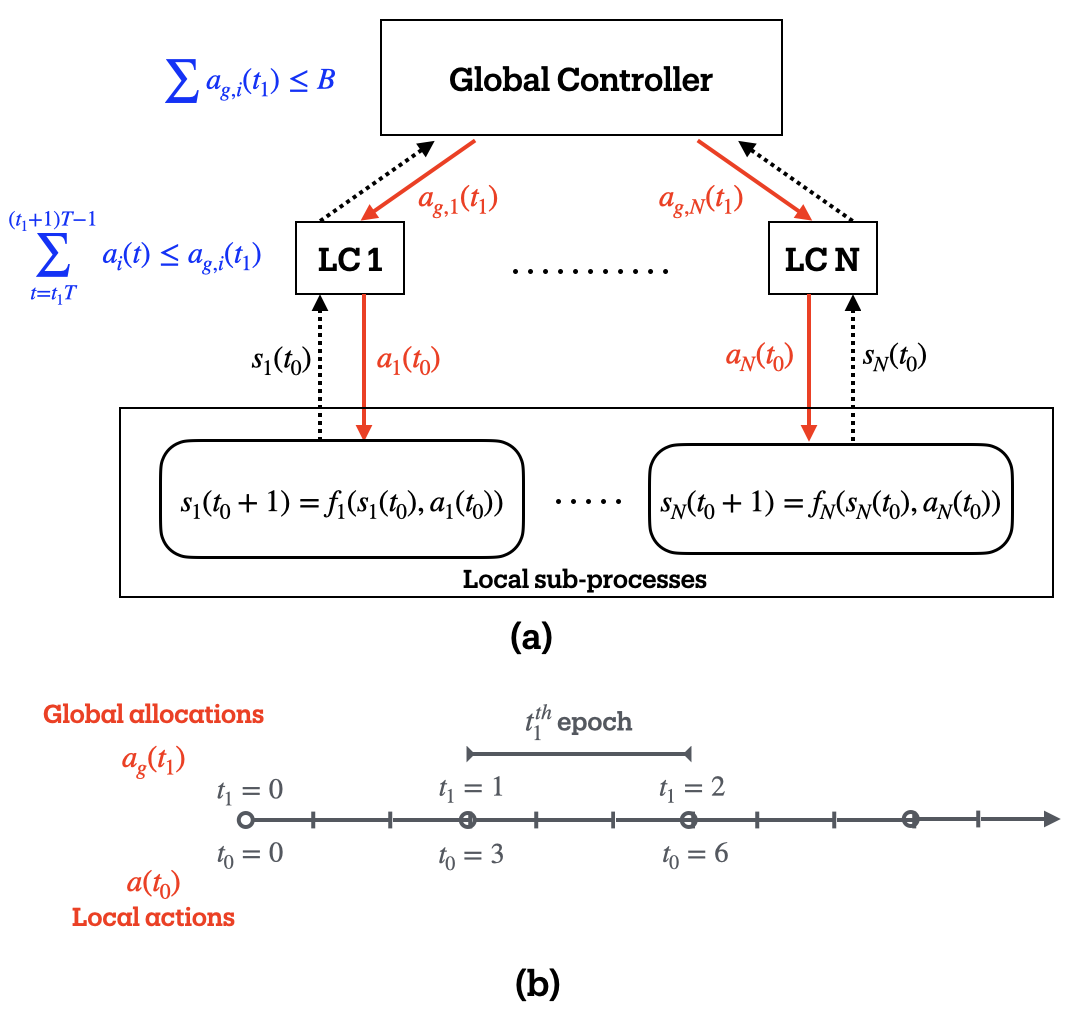}

\end{center}
\caption{(a) The hierarchical, multi-agent, multi-timescale model considered in this paper. (b) The timeline for a two timescale decision making with $T = 3$.}
\label{fig:timeline}
\end{figure}

\subsection{System Description}
\textbf{Sub-process:} The system contains \(N\) independent sub-processes \(S_i\), each associated with an LC, and each modeled as an MDP $(\mcS_i, \mcA_i, P_{i,a_i}, r_i)$, where $\mcS_i$ is its state space with cardinality $n_i$ and the set $\mcA_i$ is the action space with cardinality $ m_i$. 

Given a state-action pair $(s_i,a_i) \in \mcS_i \times \mcA_i$, $p_{i,a_i}^{lc}(s_i,s'_i)$ is the transition probability from state $s_i$ to $s'_i$ under action $a_i$ defined as 
\[
\small{ p_{i,a_i}^{lc}(s_i,s'_i) = \Pr(s_i(t_0+1) = s_i'| s_i(t_0) = s_i, a_i(t_0) =a_i))}.
\]
These probabilities form the entries of the transition matrix $P_{i,a_i} \in \mathbb{R}^{n_i \times n_i}$. $r_i:\mcS_i \times \mcA_i \to \Rr_+$ is the local reward for a state action pair $(s_i,a_i)$. 
    
\textbf{Global Controller:}
\Ram{The GC is modeled as an MDP $(\mcS^N, \mcA_G, \mathbf{P}, R)$,  where \(\mcS^N: \mcS_1\times \cdots \times \mcS_N\) is the cartesian product of all the local state spaces and the state of GC is denoted by $\gs = [s_1,\dots,s_N]$. Given a state $\gs(t_1)$, the GC decides an allocation $a_{g,i} \in \mcA_G$ for each sub-process, where $\mcA_g$ is the allocation set\footnote{To differentiate the control actions of local and global controllers, we use the term action for local controller and allocation for global controller} $\{0,1,\dots,K-1\}$, where $K \in \mathbb{Z}_+$.  The vector $a_g =[a_{g,1},\dots,a_{g,N}]$ contains the allocation for all the sub-processes.}  

\textbf{Global Policy:} At each epoch $t_1$, the global controller decides the allocation $a_g(t_1)$ according to policy $\Phi \ni \phi: \mcS^N \to \mcA_G$ such that $a_g(t_1) = \phi(\gs(t_1)).$
    
  \textbf{Local Policy:} In the $t_1^{th}$ decision epoch of the global controller, the $i^{th}$ local controller decides a sequence of $T$ control actions denoted by  
    $\mathbf{a}_i(t_1) = \{a_i(Tt_1), \dots, a_i(T(t_1+1)-1)\} $, where each $a_i(Tt_1+t) \in \mcA_i$. This is according to decision rule $\Pi_i \ni \pi_{i,t}: \mcS_i \to \mcA_i$ such that 
    \[
     a_i(Tt_1+t) = \pi_{i,t}(s_i(Tt_1+t)) \quad t = 0,\dots,T-1.
     \]
     We define the local policy as the sequence ${\boldpi}_i := \{\pi_{i,0},\dots,\pi_{i,T-1}\}$ and denote $\boldpi := [\bpi_1,\dots,\bpi_N]^T$, the vector of $T$-horizon local policies of all the subsystems, and we denote the space of all possible $\bpi$ as $\boldPi$.
     
\textbf{Global Allocation Constraint:} The global controller is constrained by an allocation budget \Ram{per epoch} \(B\), such that $\sum a_g(t_1) \le B$, $\forall\  t_1$ and $B < N(K-1)$. 
   
\textbf{Local Constraints:} The actions of the $i^{th}$ local controller over the $t_1^{\text{th}}$ epoch are constrained by the budget $a_{g,i}(t_1)$ such that $\sum_{t_0=t_1 T}^{(t_1+1)T-1}  a_i(t_0) \leq a_{g,i}(t_1)$. 
    
     The transition from a state $\gs(t_1) := \gs$ to $\gs(t_1+1) := \gs'$ under an allocation $a_g$ and local policy $\bpi$ is governed by the transition probability over an epoch, defined as 
\begin{align}
\label{eq:TP_EqMDP}
p^{ep}_{a_g,\bpi}(\gs,\gs') = Pr(\gs(t_1+1) = \gs' | \gs(t_1) = \gs, a_g, \bpi).
\end{align}

    In the above system definitions, we search for local policies $\pi_i$ that are non-stationary in the faster time scale $t_0$ (i.e., $\pi_{i,t} \neq \pi_{i,t^'}, t \neq t^'$) and stationary with respect to the slower time scale \(t_1\). Further, the global allocation policy $\phi$ is assumed to be stationary with respect to $t_1.$ The normed linear space of bounded real-valued functions over the state space $\mcS$ is denoted by $\mcV$.
     
\subsection{Problem Formulation}
Given an initial state $\gs \in \mcS^N$, the value of that state under a global allocation policy $\phi$ and local policies $\bpi$ is defined as the expected cumulative reward expressed by the following function
\small{\begin{align}
\label{eq:V_Def}
V^{\phi,\bpi}(\gs) = \begin{aligned} \Expect_{\phi,\bpi} \left[ \sum_{t_1=0}^{\infty} \beta^{t_1} R\bigg(\gs(t_1), \phi(\gs(t_1)), \bpi(\gs(t_1)) \bigg)
 \right],
\end{aligned}
\end{align}}
\normalsize{}where $R$ is the global reward function for one epoch. This is defined to be the value function of the state $\gs$. The reward function $R$ takes the following form
\small{\begin{align}
\label{eq:R_Def}
R(& \mathbf{s}(t_1),\phi(\gs(t_1)), \bpi(\gs(t_1))) = I_g(\mathbf{s}(t_1), \phi(\gs(t_1))) \nonumber\\ &+  \sum_{i=1}^{N}
     \Expect_{s_i,\pi_i} \left[ \sum_{t=0}^{T-1}\gamma^{t} r_i(s_i(t_1T+t),a_i(t_1T+t)) \right], 
\end{align}}
\normalsize{}
where 
$
a_g(t_1) = \phi(\gs(t_1)) \ ,\ a_i(t_1T+t) = \pi_{i,t}(s_i(t_1T+t)). 
$

Here, $I_g$ is the immediate global reward for a $(\gs,a_g)$ pair, and $r_i$ is the local reward function of the $i^{\text{th}}$ sub-process at each time slot. $\beta$ and $\gamma$ are the discount factors corresponding to the global and local controllers and are in the interval $(0,1)$. 
Further, it is reasonable to assume that when there is no allocation to a sub-process during an epoch, its expected reward for that is epoch is $0$, thereby imposing the following assumptions on the reward functions 
$
R(\cdot,\mathbf{0},\cdot) = 0   \mbox{ and } r_i(\cdot,0) = 0.
$

\section{Optimization frameworks}
\label{sec:Optframeworks}
Given the cumulative discounted reward $V^{\phi,\bpi}(\gs)$ as in \eqref{eq:V_Def}, we propose two optimization frameworks for computing the policies $\phi$ and $\bpi$. These frameworks fundamentally differ in the way the optimization problems for the local policies $\bpi$ are formulated. The Federal Optimization framework (abbreviated as FOpt) uses the local policy that maximizes the cumulative local reward of each sub-process over a single epoch, while the Central Optimization framework (abbreviated as COpt) computes the local policies that optimizes the infinite-horizon discounted reward. The FOpt framework, by its formulation, provides more autonomy to the local controllers in the control hierarchy, as the optimization for local policies involves maximizing the rewards corresponding to that specific sub-process alone. We formally define the two frameworks and provide an illustrative example explaining how these frameworks can result in different policies for the same system. 

\subsection{Federal Optimization Problem (FOpt)} 
Given the expected cumulative reward as in \eqref{eq:V_Def}, we define the \textit{federal} optimization problem (FOpt) for computing the policies $\phi$ and $\bpi$ as follows:
In this framework, the objective of the global controller is to compute policy $\phi$ such that 

\begin{align}
\label{eq:ValFunD}
\small{\begin{aligned}
V_F^*(\gs) &= \max_{\phi \in \Phi} \Expect_{\gs,\phi} \left[ \sum_{t_1=0}^{\infty} \beta^{t_1} R\bigg(\gs(t_1), \phi(t_1), \bpi_\phi(\gs(t_1))\bigg)\right] \\
s.t.\ \  &\ \ \mathbf{1}^T \phi(\gs(t_1)) \leq B, \forall\ t_1,
\end{aligned}}
\end{align}
where $\bpi_\phi = [\pi^\phi_{1},\dots,\pi^{\phi}_{N}]$ are the local optimal policies for a given global allocation policy $\phi$ defined through the following local optimization problem
\small{\begin{align}
\label{eq:LocOptPolicy}
\pi_{i}^\phi(s_i) &= \argmax_{\pi_i \in \Pi_i} \Expect_{s_i,\pi_i} \left[ \sum_{t=0}^{T-1}\gamma^{t} r_i(s_i(t_1T+t),{\pi_{i,t}}(s_i(t_1T+t))) \right] \nonumber \\
& \mbox{s.t.} \sum_{t=0}^{T-1} \pi_{i,t}(s_i(t_1T+t)) \leq e_i^T \phi(\gs(t_1)), 
\end{align}}
\normalsize{}
where $e_i$ is the $i^{th}$ unit vector. The optimal global allocation policy $\phi^*_F$ is defined as 
\begin{align}
\label{eq:FedGlobal}
\phi_F^* = \argmax_{\phi \in \Phi} V_F^*,
\end{align}
where $V_F^*$ is defined as in equation \eqref{eq:ValFunD} and the local optimal policy $\bpi_F^*$ is given by $\pi_{F,i}^* = \pi_{\phi_F^*,i}$, where $\pi_{\phi_F^*,i}$ is defined as in equation \eqref{eq:LocOptPolicy} for the optimal global allocation $\phi^*_F$. 

\Ram{It can be seen that local policies of the FOpt framework maximizes the reward 
\begin{align}
\label{eq:LocalReward}
R_i(s_i,\pi_i^{\phi}) := \Expect_{s_i,\pi_i^{\phi}}\left[\sum_{t=0}^{T-1} \gamma^t r_i(s_i(t),\pi_{i}^{\phi}(s_i(t)))\right]
\end{align}
 subject to the constraints in \eqref{eq:LocOptPolicy}. From the view of the global controller, the optimal local policies $\pi_{\Phi_F^*,i}$  maximizes $R_i(s_i,\pi_i^{\phi_F^*})$ over a single epoch and can be viewed as a myopic policy (over that epoch). Given any global allocation $\phi$, we formally define the local myopic policy over an epoch as the $T$-myopic policy. } 
\Ram{
\begin{definition}[$T$-myopic policy]
A local policy $\pi_i^{\phi}$ under an allocation $\phi$ is said to be $T$-myopic if it chooses actions $\mathbf{a}_i$ that maximizes the immediate expected local reward $R_i$ as defined in \eqref{eq:LocalReward} of that epoch.
\end{definition}}

\subsection{Central Optimization Problem (COpt)}
A natural alternative for the federal framework considers a single controller that determines both the global allocation and local policies together as a single optimization problem. We term it as the \textit{central} optimization problem (COpt) where the overall objective is to compute $\phi,\bpi$ such that

{\small{\begin{align}
\label{eq:ValFunCent}
V_C^{*}(\gs) &:= \max_{\phi \in \Phi}\max_{\bpi \in \Pi} V^{\phi,\bpi}(\gs) \nonumber \\ &=  
\max_{\phi \in \Phi}\max_{\bpi \in \Pi} \Expect_{\gs,\phi,\bpi} \left[ \sum_{t_1=0}^{\infty} \beta^{t_1} R\bigg(\gs(t_1), \phi(\gs(t_1)), \bpi(\gs(t_1))\bigg)\right] \nonumber \\ 
    & \mbox{s.t.} \begin{cases}  \mathbf{1}^T \phi(\gs(t_1)) \leq M , & 
     \forall\ t_1,  \\
      \sum\limits_{t=t_1 T}^{(t_1+1)T-1} \pi_i(s_i(t)) \  \leq \ e_i^T \phi(\gs(t_1))  , & \forall \ i,\ \forall \ t_1, 
     \end{cases}
\end{align}} }
\normalsize{} \Ram{where $R$ is defined as in \eqref{eq:R_Def}}. The central optimal allocation policy $\phi_C^*$ and local policies $\bpi_C^*$ are defined as 
\begin{align}
\{\phi_C^*,\bpi_C^*\} = \argmax V^*_C.
\end{align}
where $V_C^*(\gs)$ is defined as in equation \eqref{eq:ValFunCent}.

\subsection{Non-equivalence of Central and Federal frameworks}
\Ram{The two problem formulations fundamentally differ in the way the optimal policy is computed. In the COpt, the optimal policy $\phi_C^*, \bpi_C^*$ are computed as solutions of a nested infinite horizon optimization problem, while the FOpt allows the local controllers to compute their own policy $\pi_{F,i}^*$ as a solution to a finite horizon optimization problem, leading to a local and distributed computation of the policies $\bpi_F^*$, thereby providing autonomy to the local controllers. The above two frameworks lead to different optimal decision rules, making them two distinct non-equivalent stochastic control problems as illustrated in the following example.} 

\Ram{Consider the simple case of the system with a single sub-process (N=1) operating at a single timescale, i.e. $T=1$ and $t_0\equiv t_1$. Let the state space of the sub-process be $\mcS=\{0,1\}$, and the action space be $\mcA = \{0,1\}$. Considering the budget per epoch to be $1$ (i.e. $B=1$) and an equality constraint on the global allocation (i.e. $\sum a_g(t_1) = 1$), and the fact that there is only one sub-process, the allocation $a_g$ at any arbitrary time $t_1$ is trivial, i.e. $a_g(t_1)=1$. So the only computation required is for the optimal local action $a_1(t_1)$. Under FOpt, the local optimal policy in \eqref{eq:LocOptPolicy} maximizes the immediate reward at each time $t_1$, leading to a greedy/myopic policy. Assume $I_g=0$ and the local reward function to be $r_1(s_1,a_1) = \begin{bmatrix} 0.25 & 0.5 \\ 0.75 & 1 \end{bmatrix}$ for $s_1\in\mcS$ and $a_1\in\mcA$.} 

\subsubsection{Example 1 - Non-equivalent case}
 \Ram{Let the transition probability matrices be chosen as $P_{a_1=0} = \begin{bmatrix} 0.2 & 0.8 \\ 0.2 & 0.8 \end{bmatrix}$, $P_{a_1=1} = \begin{bmatrix} 0.8 & 0.2 \\ 0.4 & 0.6 \end{bmatrix}$.  
 The optimal local policy under COpt is the solution the infinite horizon MDP in \eqref{eq:ValFunCent}. Here, the value function can be computed using the value iteration algorithm (see \cite[Chapter 6]{puterman2014markov}) to be $[74.5, 75.1]^T$, and the corresponding optimal policy is computed to be $\pi_C^*(s_1=0) = 0,$ $\pi_C^*(s_1=1)=1$. As mentioned earlier, the local optimal policy under FOpt is the myopic policy which just maximizes the immediate reward, and is computed to be $\pi^*_F(s_1=0)=1,$ $\pi^*_F(s_1=1)=1$. The corresponding value function under this policy can be computed through Monte Carlo simulation, by averaging over randomly generated state and action sequences using the transition probability matrices, while the using the policy $\pi^*_F$. This value function is computed to be $[66.3,67.2]^T$. } 

\subsubsection{Example 2 - Equivalent case}
\label{subsubsec : Ex2 - Equivalent case}
\Ram{ However, for different values of transition probabilities, the two frameworks can result in the same optimal policies of both local and global controllers. For the system discussed above, consider the transition probability to be $P_{a_1=0} = \begin{bmatrix} 0.5 & 0.5 \\ 0.5 & 0.5 \end{bmatrix}$, $P_{a_1=1} = \begin{bmatrix} 0.2 & 0.8 \\ 0.2 & 0.8 \end{bmatrix}$. Under these probabilities, 
it can be seen that both FOpt and COpt frameworks lead to the same policies at the local level with $\pi^*(s_1=0)=1$ and $\pi^*(s_1=1)=1$. The value function under this policy is computed to be $[89.60, 90.08]$. Kindly note that the optimal global allocation is $a_g(t_1) = 1$ by design.}

\section{Analysis and Results}
\label{sec: analysis}
\Ram{The above example illustrates that these two frameworks are non-equivalent by formulation. Additionally, the two-time scale formulation makes it non-trivial from the existing MDP literature and necessitates the investigations on the existence of optimal value functions and their uniqueness for COpt and FOpt frameworks.} In this section, we prove the existence and uniqueness of the optimal value functions, define value iteration algorithms, and provide comparative results between the two frameworks and specifically bound the difference between the two value functions. 

\Ram{The initial results correspond to the existence of the optimal global and local policies for the COpt and FOpt formulations. In particular, we define the Bellman operator on the space of value functions and prove it to be a contraction mapping, leading to the existence of a fixed point that corresponds to the optimal value function.}

\begin{theorem}
\label{th:COpt}
For all $\gs\in \mcS^N$, there exists an unique $V_C^* \in \mcV$ satisfying
\begin{align}
\label{eq:COpt_Bellman}
\begin{aligned}
V_C^*(\gs) &= \max_{a_g \in \mcA_g^B} \max_{\bpi \in \mathbf{\Pi}_{a_g}} \bigg\lbrace  R(\gs,a_g,\bpi(\gs))   \\ &   +\beta \sum_{\gs' \in  \mcS^N} p^{ep}_{a_g,\bpi}(\gs,\gs{'}) V_C^*(\gs')   \bigg\rbrace.
\end{aligned}
\end{align}
\end{theorem}
\begin{proof}
Equation \eqref{eq:COpt_Bellman} can be viewed as a nested optimization problem on global allocation $a_g$ and local policy $\bpi$. For a pair of $a_g, \bpi$, the reward function $R(\gs,a_g,\bpi(\gs))$ is bounded, and $\beta < 1$. The Bellman operator $\Bel_C: \mcV \to \mcV$ for the COpt formulation is given by
\begin{align}
\begin{aligned}
\Bel_C V_C(\gs) = \max_{a_g \in \mcA_g^B}& \left\lbrace \max_{\bpi \in \mathbf{\Pi}_{a_g}} \Bigg[ R(\gs,a_g,\bpi(\gs))
\right. \\  & \left. +\beta \sum_{\gs' \in  \mcS^N} p^{ep}_{a_g,\bpi}(\gs,\gs{'}) V_C(\gs')  \Bigg] \right\rbrace.
\end{aligned}
\label{eq:BEL_COpt}
\end{align}

Since the state space $\mcS$ is finite and $\beta < 1$, \cite[Proposition 6.2.4 and Theorem 6.2.5]{puterman2014markov} guarantee that $\Bel_C$ is a contraction mapping on $\mcV$ and hence a fixed point $V_C^*$ exists.
\end{proof}
\begin{theorem}
\label{th:FOpt}
For all $\gs\in \mcS^N$, there exists an unique $V_F^* \in \mcV$ satisfying
\begin{align}
\begin{aligned}
V_F^*(\gs) &= \max_{a_g \in \mcA_g}\bigg \lbrace R(\gs,a_g,\bpi^{a_g}_{F}) \\ &  +\beta \sum_{\gs' \in  \mcS^N} p^{ep}_{a_g,\bpi_F^{a_g}}(\gs,\gs{'}) V_F^*(\gs')   \bigg\rbrace.
\end{aligned}
\label{eq:Fopt_Bellman}
\end{align}
\normalsize{}
where $\bpi^{a_g}_F$ is the FOpt policy corresponding to the allocation $a_g$.
\end{theorem}
\begin{proof}
For every allocation $a_g$ of the GC, the local optimization \eqref{eq:LocOptPolicy} is a finite horizon MDP with finite state and action spaces,  \cite[Proposition 4.4.3]{puterman2014markov} guarantees the existence of a deterministic local optimal policy $\bpi^{a_g}_F$ and it can be computed using any policy evaluation algorithm. Given the optimal local policies $\bpi_F^{a_g}$, the reward function $R(\gs,a_g,\bpi^{a_g}_F)$ is unique and bounded for every $(\gs,a_g)$ pair. The Bellman operator $\Bel_F$ for the FOpt problem can be defined as 
\small{\begin{align}
\label{eq:Bel_FOpt}
\begin{aligned}
\Bel_F V_F(\gs) = \max_{a_g \in \mcA_g^B}  &\Bigg[ R(\gs,a_g,\bpi_F^{a_g}(\gs)) \\ 
&+\beta \sum_{\gs' \in  \mcS^N} p^{ep}_{a_g,\bpi_F^{a_g}}(\gs,\gs{'}) V_F(\gs')  \Bigg]. 
\end{aligned}
\end{align}}
\normalsize{}
Additionally, given $\beta < 1$, and the state space $\mcS$ is finite, from  \cite[Proposition 6.2.4, Theorem 6.2.5]{puterman2014markov}, we claim that $\Bel_F$ is a contraction mapping on $\mcV$ and guarantee the uniqueness of $V_F^*$.
\end{proof}
\begin{remark}
Rewriting $\Bel_F$ as
\small{\begin{align*}
\Bel_F V_F(\gs) = \max_{a_g \in \mcA_g^B}& \left\lbrace \max_{\bpi \in \mathbf{\Pi}_{a_g}} \Bigg[ R(\gs,a_g,\bpi(\gs))
\Bigg] \right. \\  & \left. +\beta \sum_{\gs' \in  \mcS^N} p^{ep}_{a_g,\bpi_F^{a_g}}(\gs,\gs{'}) V_F(\gs')  \right\rbrace,
\end{align*}}
\normalsize{}
we immediately see the difference between the COpt \eqref{eq:BEL_COpt} and FOpt formulations, where the former is optimizing both the local and global policies over infinite horizon, while the later optimizes only the global policy over infinite horizon.
\end{remark}
\begin{remark}
From Theorems \ref{th:COpt} and \ref{th:FOpt}, a value iteration algorithm for computing the optimal value functions can be formulated as 
\begin{align}
\label{eq:ValIter}
\begin{aligned}
V_C^{(m+1)}(\gs) &= \Bel_C V_C^{(m)}(\gs), \\ 
V_F^{(m+1)}(\gs) &= \Bel_F V_F^{(m)}(\gs).
\end{aligned}
\end{align}
Starting from initial values $V_C^{(0)}$ and $V_F^{(0)}$, from Theorem 6.3.1 \cite{puterman2014markov}, we have $V_C^{(m)}  \to V_C^*$ and $V_F^{(m)} \to V_F^*$. 
\end{remark}

\Ram{Given that the two optimization frameworks are well defined, leading to unique value functions, the following results focus on the comparative analysis between the two. In particular, we} bound the difference between the two value functions (i.e., $||V_C^*(\gs)-V_F^*(\gs)||$). \Ram{This can be viewed as the \emph{cost of autonomy}, the cost incurred for using a federal structure rather than a centralized decision framework.} The following lemma gives the relationship between the two value functions.
\begin{lemma}
\label{lem:Dec-Cent1}
Given an optimal policy $\{\phi_C^*, \bpi_C^*\}$ of COpt and an optimal policy $\{\phi_F^*, \bpi_F^*\}$ of FOpt, we have
\[
V^*_F(\gs) \leq V^*_C(\gs).
\]
\end{lemma}

\begin{proof}
 From the definition of FOpt \eqref{eq:ValFunD}, we have
\small{\begin{align*}
V^*_F(\gs) &= \max_{\phi \in \Phi} \Expect_{\gs,\phi} \left[ \sum_{t_1=0}^{\infty} \beta^{t_1} R(\gs(t_1), \phi(\gs(t_1)), \bpi_\phi(\gs(t_1)))\right] 
\end{align*}
\begin{align*}
&\leq \max_{\phi \in \Phi} \max_{\bpi \in \Pi} \Expect_{\gs,\phi} \left[ \sum_{t_1=0}^{\infty} \beta^{t_1} R(\gs(t_1), \phi(\gs(t_1)), \bpi(\gs(t_1)))\right] \\ 
& = V_C^*(\gs).
\end{align*}}
\normalsize{}
The inequality is due to the fact that the COpt framework searches over a larger space $\phi,\bpi \in \Phi \times \Pi$ while the FOpt framework searches over $\phi, \bpi_\phi$, where $\bpi_\phi$ is fixed for a given $\phi$ according to equation \eqref{eq:LocOptPolicy}, thereby making the search space for policies constrained than the COpt framework.
\end{proof}

\begin{theorem}
Given the value functions $V^*_F(\gs)$ and $V^*_C(\gs)$ and the corresponding optimal policies $\{\phi_F^*,\bpi_F^*\}$ and $\{\phi_C^*,\bpi_C^*\}$ of the FOpt and COpt problems, we have
\begin{align}
\label{eq:TempT31}
V^{\phi_C^*,\bpi_{\phi_C^*}}(\gs) \leq V_F^*(\gs) \leq V_C^*(\gs).
\end{align}
where $\bpi_{\phi_C^*}$ is the solution of the local optimization problem (\ref{eq:LocOptPolicy}) under the optimal global allocation policy $\phi_C^*$ of the COpt problem. 
Additionally,
\begin{align}
||V_C^*(\gs) - V_F^*(\gs) || \leq || V_C^*(\gs)-V^{\phi_C^*,\bpi_{\phi_C^*}}(\gs)||.
\end{align}

\end{theorem}
\begin{proof}
From Lemma \ref{lem:Dec-Cent1} we know that $V_F^*(\gs) \leq V_C^*(\gs)$. To get the other side of the required relation, given the optimal global allocation policy $\phi_C^*$ of COpt, we know that $V^{\phi_C^*,\bpi}(\gs) \leq V_C^*(\gs)$ for all $\bpi \neq \bpi_C^*$.
In particular, we choose the local policy to be $\bpi_{\phi_C^*}$, the solution of the local optimization problem \eqref{eq:LocOptPolicy} given the global allocation policy $\phi_C^*$. If $V_F^*(\gs) < V^{\phi_C^*,\bpi_{\phi_C^*}}(\gs)$, then it leads to a contradiction as $V_F^*(\gs)$ is the maximum achievable value function of FOpt than any other policies that are computed through the local optimization problem \eqref{eq:LocOptPolicy}, such as the policy $\{\phi_C^*, \bpi_{\phi_C^*}\}$. So, $V^{\phi_C^*,\bpi_{\phi_C^*}}(\gs)$ should have a value function which is upper bounded by $V_F^*(\gs)$. Hence, we have 
\begin{align}
\label{eq:temp2_T1}
V^{\phi_C^*,\bpi_{\phi_C^*}}(\gs) \leq V_F^*(\gs).
\end{align}
This above equation and Lemma \ref{lem:Dec-Cent1} proves the claim in equation \eqref{eq:TempT31}. Taking the norm of the differences between the terms, we have the second statement.
\end{proof}
In the above result, we bound how far FOpt can be from COpt by computing the $T$-myopic local policy for $\phi_C^*$. A major consequence of the above theorem is that it gives conditions under which FOpt and COpt are equivalent. 
\begin{corollary}
\label{cor:equivalence}
FOpt and COpt have the same value function iff the local policies $\pi_C^*$ of the COpt are $T$-myopic. 
\end{corollary}

\section{Equivalence conditions for COpt and FOpt}
\Ram{Though Corollary \ref{cor:equivalence} completely characterizes the conditions for equivalence of the FOpt and COpt frameworks, it requires the computation of the optimal policies to verify the same. It is desirable to develop structural conditions on the system that can be verified without explicitly computing the optimal policies to guarantee the equivalence. In this section, we proceed to establish sufficient conditions on the system, specifically on the transition probability and reward functions for the equivalence. For this, we impose the following assumptions on the system.} 
\begin{itemize}
\item[(A1)] The sets $\mcS_i$ and $\mcA_i$ are partially ordered with ordering $\leq_{\mcS_i}$ and $\leq_{\mcA_i}$. The ordering on $\mcS_i$ inherently defines an ordering $\leq_\mcS$ on the global state space $\mcS^N$. 
\item[(A2)] The local reward function $r_i$ is non-decreasing in both $s_i$ and $a_i.$ Further, the function $I_g$ is non-decreasing in $\gs$.
\end{itemize}
\Ram{Such an ordering on state and action spaces is necessary to discuss the monotonicity of the value functions \cite{puterman2014markov}. Further, we impose the following structural assumptions on the state transition probabilities.}
{
\begin{itemize}
\item[(A3)] $p^{ep}_{a_g,\bpi^*_{a_g}}(\gs) \coloneqq [p^{ep}_{a_g,\bpi^*_{a_g}}(\gs,\gs')]_{\gs'\in\mcS}$ is stochastic monotone in $\gs,$ i.e., $p^{ep}_{a_g,\bpi^*_{a_g}}(\gs)\geq_{st} p^{ep}_{a_g,\bpi^*_{a_g}}(\gs')$ for $\gs \geq_{\mcS} \gs^{'}$.  
\item[(A4)] For $\bpi_{a_g}^* \neq \bpi_{a_g}$, $p^{ep}_{a_g,\bpi_{a_g}^*}(\gs)  \geq_{st} p^{ep}_{a_g,\bpi_{a_g}}(\gs)$.
\item[(A5)] For any $a_i,$ $P^{lc}_{i,a_i}(s_i)\coloneqq [p^{lc}_{i,a_i}(s_i,s_i')]_{s_i'\in\mcS_i} $ is stochastic monotone in $s_i$, i.e. $P^{lc}_{i,a_i}(s_i)\geq_{st} P^{lc}_{i,a_i}(s_i')$ for $s_i\geq_{\mcS} s_i'$. 
%
\end{itemize}
}
{
\begin{definition}
Stochastic ordering: To probability mass functions $p$ and $q$ with $n$ components are said to be stochastically ordered as $p\geq_{st}q$, if \[
\sum_{j\geq i} p(j) \geq \sum_{j\geq i} q(j) \ \ \forall i=1,\dots,n.
\]
Equivalently, for any bounded, increasing function $g:\mathbb{R}\to\mathbb{R}$, 
\[
\sum_{i=1}^{n}p(i)g(i) \geq \sum_{i=1}^{n}q(i)g(i).
\]
\end{definition}
}
 \Ram{We refer to \cite[Section 4.7.3]{puterman2014markov} for a more formal definition of stochastic ordering ($\geq_{st}$), which is used to compare two random variables (i.e. their corresponding probability distributions). Intuitively, given a state, if we consider two transition probability vectors, the dominant of the two vectors will lead to ``higher'' states with greater probability. This means that the dominant vector leads to a greater future expected cumulative reward if the value function is monotonically increasing in state.}
 
\Ram{The following theorem guarantees that the above assumptions are sufficient for the equivalence of FOpt and COpt frameworks. }

\begin{theorem}
\label{thm:Sufficiency}
Under assumptions (A1) to (A5), value functions $V_F^*(\gs)$ and $V_C^*(\gs)$ of FOpt and COpt are equal $\forall \mathbf{s}\in\mcS^N$. 
\end{theorem}

\Ram{The proof of the theorem requires some intermediate results, which are developed in the following lemmas. First, Lemma~\ref{lemma:R_monotone} shows that the global reward function $R(\gs,\phi(\gs),\pi(\gs))$ defined in \eqref{eq:R_Def} is monotone in the state $\gs$. Then, we further show that the value function $V_C^*(\gs)$ as in \eqref{eq:ValFunCent} of COpt problem is monotone in the state $\gs$ in Lemma~\ref{lemma:V_C^* monotone}. Using this monotonicity property of the value function and stochastic ordering property of the transition probabilities, the proof of Theorem~\ref{thm:Sufficiency} is established.}

\begin{lemma}
\label{lemma:R_monotone}
\Ram{Given $a_g = \phi(\gs)$, and the corresponding local optimal policy $\bpi^*_{a_g}$ defined as in \eqref{eq:LocOptPolicy}, assumptions (A1) to (A5) are sufficient for the global reward function $R$ in \eqref{eq:R_Def} to be non-decreasing over $\mcS^N$.} 

\end{lemma}
\begin{proof}\normalsize{}
Consider the following optimal local value function, 
\begin{align*}
R_i(s_i,\pi_i^{a_g}) &=  
\max_{\pi_i} \Expect_{s_i,\pi_i} \left[ \sum_{k=0}^{T-1}\gamma^{k} r_i(s_i(k),a_i(k)) \right], \\ 
& s.t. \ \ \sum_{k=0}^{T-1}a_i(k) \leq a_{g,i}, \ \  \forall i=1,2,\cdots,N.
\end{align*}
 We will show that it is non-decreasing in $s_i.$ Monotonicity results of unconstrained finite horizon MDPs cannot be applied directly to the constrained MDP. Hence we reformulate the constrained local MDP into an equivalent unconstrained $T$-horizon MDP, with an augmented state $(s_i(t),b_i(t))$ for $t\in \{0,1,\ldots,T-1\}$, where, $b_i(t)$ is the remaining budget available at time $t$. This optimal value function of this MDP satisfies the following equations (some subscripts are dropped for convenience). 
{\small{
\begin{align}
    v_{i,t}^*(s,b) &= \max_{a\in \mcA_{i}(b)} q_{i,t}^*(s,b,a),  \\
    q_{i,t}^*(s,b,a) &\coloneqq \left\lbrace r_i(s,a) + \gamma\mathbf{E}\left[ v_{i,t+1}^*(s_i(t+1),b-a)|s, b, a)\right]\right\rbrace.\nonumber
\end{align} }}
\normalsize{}
Here, $\mcA_{i}(b) = \{a\in\mcA_i|a\leq b\}$, the set of actions within the available budget. 
For a given action $a$, assuming $v_{i,t+1}^*$ is non-decreasing in $s$, the quantity $\mathbf{E}\left[ v_{i,t+1}^*(s_i(t+1),b-a)|s, b, a)\right]$ is non-decreasing in $s$ due to the stochastic monotone property of transition probability vector $P_{s}^{lc}(i,a_i)$ in $s$ (from assumption (A5)). For each $t$, $q_{i,t}^*(s,b,a)$ is non-decreasing in $s$ because it is a sum of two non-decreasing quantities. 
Applying the principle of induction, it can be shown that $v_{i,t}^*$ is non-decreasing in $\mcS_i$ for a given budget for every $t=0,\ldots,T-1$. 

Further, the sum of value functions across all LCs is non-decreasing over $\mcS^N$. 
Additionally, the global reward function $I_g$ is non-decreasing in $\mcS^N$, which proves that $R$ is non-decreasing in $\mcS^N$. 
\end{proof}
\begin{lemma}
Under assumptions (A1) to (A5), $V_C^*$ is non-decreasing in $\gs$.
\label{lemma:V_C^* monotone}
\end{lemma}
\begin{proof}
 From Remark 2 we know that, $V_C^*$ can be computed through a Bellman update of $V_C^{(m)}$ at each step $m=1,2,...$ until convergence, starting from $V_C^{(1)}=R$. The Bellman operation is monotone \cite{puterman2014markov} in $m$ and we have $V_C^{(m+1)}(\gs)=\mathcal{B}_C V_C^{(m)} (\gs)\geq V_C^{(m)} (\gs)$ for all $\gs$. By induction, we will prove that $V_C^*$ is monotonic in $\gs$. We know from Lemma~\ref{lemma:R_monotone} that $V_C^{(1)}$ is non-decreasing in $\gs$.  Assuming $V_C^{(m)}$ is monotonic in $\gs$, $V_C^{(m+1)} = \mathcal{B}V_C^{(m)}$ is computed as 
{\small{
\begin{align*}
&V_C^{(m+1)}(\gs)  \\
&= \max_{a_g}\max_{\bpi} \left\lbrace R(\gs, a_g, \bpi(\gs)) + \beta \mathbf{E}\left[ V_C^{(m)}(\gs')|\gs, a_g, \bpi(\gs)\right] \right\rbrace , \\
&= \max_{a_g} \left\lbrace R(\gs, a_g, \bpi^*(\gs)) + \beta \mathbf{E}\left[ V_C^{(m)}(\gs')|\gs, a_g, \bpi^*(\gs)\right] \right\rbrace .
\end{align*}
}}
Since, we assumed that $V_C^{(m)}$ is non-decreasing in $\gs$, we have
{\small{
\begin{align*}
\mathbf{E}\left[ V_C^{(m)}(\gs') \vert \gs'', a_g, \bpi(\gs)\right] \geq \mathbf{E}\left[ V_C^{(m)}(\gs') \vert \gs''', a_g, \bpi(\gs)\right] .
\end{align*}
}}
for $\gs''>_{\mcS} \gs'''$. This is due to assumptions (A3), (A4), and the definition of stochastic monotonicity. Hence, $\mathbf{E}\left[ V_C^{(m)}(\gs')|\gs, a_g, \bpi^*(\gs)\right]$ is non-decreasing in $\gs$. This means $V_C^{(m+1)}(\gs)$ is non-decreasing in $\gs$. 
We know from value iteration that $V_C^{(m)}(\gs)$ converges to $V^*_C$ as $m\to \infty$. Thus, by induction, $V_C^*$ is non-decreasing in $\gs$.
\end{proof}

\noindent Using the above results, we proceed to prove Theorem \ref{thm:Sufficiency}
\begin{proof}[Theorem \ref{thm:Sufficiency}]
Rewriting the value function of COpt problem \eqref{eq:COpt_Bellman}
{\small{
\begin{align*}
V_C^{*}(\gs) &= \max_{a_g}\left\lbrace R(\gs, a_g, \bpi_{a_g}^*(\gs)) + \beta \mathbf{E}\left[ V_C^{*}(\gs')|\gs, a_g, \bpi_{a_g}^*(\gs)\right] \right\rbrace .
\end{align*}
}}
\normalsize{}
Consider a $T$-myopic policy $\bpi_{a_g}^*$ that maximizes $R(\gs, a_g, \bpi_{a_g}^*(\gs))$ for a state $\gs$ and allocation $a_g$. Under assumption (A5) and monotonicity of $V_C^*(\gs')$ in $\gs'$ (Lemma~\ref{lemma:V_C^* monotone}), this policy  also maximizes the expected future value $\mathbf{E}[V_C^{*}(\gs')|\gs, a_g, \bpi_{a_g}^*(\gs)]$, and hence their sum. Hence $\bpi_{a_g}^*$ is the solution to the COpt problem. Additionally, by definition, FOpt always uses a $T$-myopic policy, leading to the equivalence between the two frameworks. 
\end{proof}
{
\begin{remark}
The transition probability in example \ref{subsubsec : Ex2 - Equivalent case} satisfy the conditions (A3-A5) and the reward functions satisfy A2, thereby providing an illustration of the Theorem \ref{thm:Sufficiency}.
\end{remark}
\begin{proof}
There is only one LC, as $N=1$.
We have $\mcS_1=\{0,1\}$, $\mcA_1=\{0,1\}$ which are ordered sets. So, (A1) is satisfied. 
For (A2), we have only one reward function $r_1(s_1,a_1) = \begin{bmatrix} 0.25 & 0.5 \\ 0.75 & 1 \end{bmatrix}$ for $s_1\in\mcS$ and $a_1\in\mcA$. In this matrix, all rows are increasing from left to right, and the columns are increasing from top to bottom. Also, $I_g=0$ for all states. So, (A2) is satisfied.  For (A3), consider the following.  $a_g=1$ for all epochs, due to the hard equality constraint. We have $\pi_{a_g}^* = [1,1]$. As $T=1$, we have $p^{ep}_{a_g,\pi^*_{a_g}}(s,s') = p^{lc}_{1,a_1=\pi^*_{a_g}(s)}(s,s') = P_{a_1=1}(s,s')$. Now, in $P_{a_1=1}$, we have $P_{a_1=1}(s=0)=[0.2,0.8]$, $P_{a_1=1}(s=1)=[0.2,0.8]$, and $P_{a_1=1}(s=1)\geq_{st}P_{a_1=1}(s=0)$. 

{For probability vectors of two dimensions, the one with a greater second component is dominant (according to the definition of stochastic ordering). 
For (A4), consider the following. We have four possible policies for $\pi_{a_g=1}$, viz. $[0,0]$ or $[0,1]$ or $[1,0]$ or $[1,1]$. We know $\pi^*_{a_g=1}=[1,1]$. Now, (A4) can be verified by first finding $p^{ep}_{a_g=1,\pi_{a_g}}$ using elements of matrices $P_{a_1=0},P_{a_1=1},$ and comparing it with $p^{ep}_{a_g=1,\pi^*_{a_g}}$. For example, for policy $\pi_{a_g=1}=[0,1]$, we have $p^{ep}_{a_g=1,\pi_{a_g}=[0,1]}(s=0) = [0.5, 0.5] \leq_{st} p^{ep}_{a_g=1,\pi^*_{a_g}=[1,1]}(s=0) = [0.2, 0.8]$. $p^{ep}_{a_g=1,\pi_{a_g}=[0,1]}(s=1) = [0.5, 0.5] \leq_{st} p^{ep}_{a_g=1,\pi^*_{a_g}=[1,1]}(s=1) = [0.2, 0.8]$. Other cases can be verified similarly.} For verifying (A5) we see that in $P_{a_1=1}$, we have $P_{a_1=1}(s=1)\geq_{st}P_{a_1=1}(s=0)$. Same is the case for $P_{a_1=0}$. So, (A5) is satisfied.
\end{proof}
}
\section{Conclusion}
We presented a hierarchical stochastic control architecture for cyber-physical systems, comprising of multiple independent sub-processes and local controllers, along with a global controller imposing budget constraints on individual local controllers. This is modeled as a two-timescale multi-MDP-based formulation, which hasn't been analyzed in the literature. We proposed two optimization frameworks FOpt and COpt, and analyzed the optimality of the value functions. Further, we show the equivalence between the two under specific structural assumptions on the model. These frameworks have applications in various domains, including budget allocation, healthcare infrastructure planning, energy systems, and coordinated control of autonomous vehicles. We plan to explore these applications in future work.

\bibliographystyle{IEEEbib.bst}
\bibliography{refsmdp}
\end{document}